\documentclass[11pt, onecolumn]{IEEEtran}

\usepackage{amsthm, amssymb}

\def\snr{\textrm{SNR}}

\newtheoremstyle{slplain}
  {3pt}
  {3pt}
  {\slshape}
  {}
  {\bfseries}
  {.}%
  { }
  {}

\theoremstyle{slplain}

\newtheorem{lem}{Lemma}
\newtheorem{pro}{Proposition}

\usepackage{cite}
\usepackage{amsfonts}
\usepackage{multicol,multienum}
\usepackage{subfigure}
\usepackage{multirow}

\ifCLASSINFOpdf
  \usepackage[pdftex]{graphicx}
  \graphicspath{{../pdf/}{../jpeg/}}
  \DeclareGraphicsExtensions{.pdf,.jpeg,.png}
\else
  \usepackage{float}
  \usepackage[dvips]{graphicx}
  \graphicspath{{../eps/}}
  \DeclareGraphicsExtensions{.eps}
\fi

\usepackage[cmex10]{amsmath}
\usepackage{algorithm}
\usepackage{algorithmic}
\usepackage{array}
\usepackage{url}
\usepackage{setspace}

\begin{document}

\title{Optimal Relay Probing in Millimeter Wave  Cellular Systems with Device-to-Device Relaying}

\author{
\IEEEauthorblockA{Ning Wei, Xingqin Lin, and Zhongpei Zhang}
\thanks{N. Wei and Z. Zhang are with the National Key Laboratory of Science and Technology on Communications,
University of Electronic Science and Technology of China.  (Email: \{wn, zhangzp\}@uestc.edu.cn.)  X. Lin is with Ericsson Research, San Jose, CA, USA. (Email: xingqin.lin@ericsson.com.) 
}
}

\maketitle

\begin{abstract}
Millimeter-wave (mmWave) cellular systems are power-limited and susceptible to blockages. As a result, mmWave connectivity will be likely to be intermittent.  One promising approach to increasing mmWave connectivity and range is to use relays. Device-to-device (D2D) communications open the door to the vast opportunities of D2D and device-to-network relaying for mmWave cellular systems. In this correspondence, we study how to select a good relay for a given source-destination pair in a two-hop mmWave cellular system, where the mmWave links are subject to random Bernoulli blockages. In such a system, probing more relays could potentially lead to the discovery of a better relay but at the cost of more overhead. We find that the throughput-optimal relay probing strategy is a pure threshold policy: the system can stop relay probing once the achievable spectral efficiency of the currently probed two-hop link exceeds some threshold. In general, the spectral efficiency threshold can be obtained by solving a fixed point equation. For the special case with on/off mmWave links, we derive a closed-form solution for the threshold. Numerical results demonstrate that the threshold-based relay probing strategy can yield remarkable throughput gains.
\end{abstract}



\section{Introduction} 

Recently, there has been a surge of interest in millimeter-wave (mmWave) cellular systems, operating in the 10-300 GHz band \cite{rappaport2013millimeter, khan2011millimeter, roh2014millimeter}.  Utilizing the large chunks of mmWave spectrum has the potential to enable the next-generation cellular systems to support multiple gigabit-per-second data rates and also can help mitigate the current spectrum crunch \cite{andrews2014will}. Inspired by the great potential, much research work has been done to address the various aspects of mmWave cellular system design, including propagation measurements and modeling, air-interface, beamforming, radio frequency components, and network architecture. We refer to \cite{rangan2014millimeter} and references therein for a comprehensive review of the up-to-date mmWave research.

One major concern for mmWave communications is that their signals are quite susceptible to blockages \cite{singh2009blockage, bai2014coverage}. In particular, mmWave signals cannot penetrate many solid materials and even human body can attenuate the signals by as much as 20 to 35 dB \cite{lu2012modeling}. This implies that mmWave connectivity is likely to be highly intermittent and the system needs to rapidly adaptable to the time-varying radio environment.  One promising approach to overcoming the drawbacks of the peculiar mmWave propagation characteristics is to use relays \cite{rangan2014millimeter}.  The intuition is that relaying can help mmWave signals turn around the blockages and increase the chance to reach the destinations. Indeed, recent studies have shown that multi-hop relaying can greatly increase mmWave connectivity \cite{singh2009blockage, lin2014connectivity}. Moreover, relaying is essential if the system targets at providing outdoor-to-indoor mmWave coverage. Considering the importance of relaying in mmWave systems, it is of great interest to explore the relaying opportunities enabled by the emerging device-to-device (D2D) communications in cellular systems \cite{lin2013overview}.

Existing research on multi-hop cellular networks has been focused on networks operating on the spectrum below 5 GHz \cite{pabst2004relay,le2007multihop,gozalvez2013experimental, lin2014multihop}. Many different relaying schemes have been proposed: analog repeater, amply-and-forward, decode-and-forward (DF), compress-and-forward, and  demodulate-and-forward \cite{drucker1988development, laneman2004cooperative, kramer2005cooperative, chen2006modulation}. Once a relaying scheme has been determined for a source-destination pair, the next key question is which relay(s) should be selected to assist the communication in the presence of multiple potential relays. Choosing multiple relays can potentially provide a higher diversity gain but requires more overhead, which may not be desirable from a system perspective \cite{shan2009distributed}. An interesting result proved in \cite{bletsas2006simple} is that selecting the best relay can achieve the same diversity-multiplexing tradeoff obtained by using multiple relays. Therefore, studying how to select a single ``best'' relay is of interest from both theoretical and practical perspectives. The definition of the best relay, however, hinges on the adopted performance metrics such as transmission rate \cite{zhou2011link}, network lifetime \cite{zhai2009lifetime}, and spatial reuse \cite{marchenko2009selecting}. Relay selection has also been jointly studied with other problems including power control \cite{wang2009distributed}, channel allocation \cite{ng2007joint}, and multi-antenna techniques \cite{fan2007mimo}. Now multi-hop cellular networks have been/are being standardized and deployed in practice \cite{yang2009relay}. In addition to infrastructure-based relaying, D2D communications bring in new opportunities for D2D and device-network relaying \cite{lin2013overview}.

In contrast to the existing multi-hop cellular research focusing on communications using the spectrum below 5 GHz, in this correspondence we focus on relay probing and selection in mmWave cellular systems. Specifically, we study a two-hop DF mmWave system. For a given source-destination pair, determining if a relaying device is good or not requires learning the channel qualities of source-relay and relay-destination channels. Note that mmWave transmission requires beamforming to overcome the high pathloss as well as other losses due to rain and oxygen absorption and higher noise floor associated with larger bandwidth \cite{hur2013millimeter}. As a result, to estimate the channel quality of a mmWave link, the transmitter and the receiver need to steer their antenna beams towards each other before carrying out the estimation. This beam searching and steering process results in additional non-negligible communication overhead. For example, the prototype in \cite{roh2014millimeter} requires $45$ ms for each adaptive beam searching and switching. So there is a tradeoff in searching for good relaying devices for mmWave systems:  Probing more relaying devices increases the probability of finding a better relay but at the cost of more probing overhead. This tradeoff naturally raises the central question studied in this correspondence: how many relaying devices should be probed in a mmWave system?  In contrast, such non-negligible beamforming overhead in relay probing is not a concern in non-mmWave multi-hop cellular systems, and to the best of our knowledge has not been studied in the literature.

In this correspondence, the mmWave links are subject to random Bernoulli blockages. Similar models have been proposed by \cite{bai2014coverage, singh2014tractable} to study the coverage and capacity of mmWave systems. We also explicitly take into account the beamforing overhead involved in the relay probing process. For such a mmWave system, we are interested in finding the best relaying device that achieves the maximum throughput (bit/s). Using optimal stopping theory \cite{ferguson2012optimal}, we show that the throughput-optimal relay probing strategy is a pure threshold policy. Specifically, the system can stop relay probing once the achievable spectral efficiency (bps/Hz) of the currently probed two-hop DF link exceeds some threshold. Further, it is not necessary to recall any previously probed relays: just select the relay probed at the stopping stage. In general, the spectral efficiency threshold can be obtained by solving a fixed point equation. For the special case with on/off mmWave links, we derive a closed-form solution for the threshold; the derived solution reveals the impact of key system parameters on the threshold and maximum throughput. We also numerically compare the threshold-based relay probing strategy to several heuristic relay probing schemes. Numerical results demonstrate that the threshold-based relay probing strategy can yield remarkable throughput gains.

\section{System Model and Problem Formulation}

\subsection{MmWave Transmission with I.I.D. Bernoulli Blocking}

A distinct feature of mmWave transmission is its sensitivity to the blockage of the spatially distributed obstacles in the radio environment. In this correspondence we consider random Bernoulli blockages, which has been proposed by \cite{bai2014coverage, singh2014tractable} to study the system level performance of mmWave systems. Specifically, denote by $x$ and $y$ the transmitter and the receiver of a typical mmWave link. Whether the  link $x\to y$ is blocked or not is modeled by a Bernoulli random variable $\chi_{x,y}$, which equals $1$ (with probability $p$) if the link is \textit{not} blocked and $0$ otherwise. To avoid triviality, we assume $p\neq 0$. The blocking events are assumed to be i.i.d. across the links. The received signal power at receiver $y$ from transmitter $x$ is then modeled as
\begin{align}
P_{r} (x,y) = \chi_{x,y} \eta_{x,y} \cdot P_t G_t G_r       \textrm{PL}( \| y - x \| ) ,
\end{align}
where $\eta_{x,y}$ models other random factors such as shadowing for the link $x\to y$, $P_t$ denotes the transmit power, $G_t$ and $G_r$ denote the transmit and receive beamforming gains respectively, and $\textrm{PL} ( \cdot )$ is a distance-dependent pathloss function. Further, $\{\eta_{x,y}\}$ are assumed to be i.i.d. across the links.

\subsection{Relaying Protocols}

We assume a half-duplex DF relaying strategy, where the time resource for data transmission is equally divided between source-relay transmission and relay-destination transmission; this work could be straightforwardly extended to other relaying strategies. The spectral efficiency (bit/s/Hz) of the two-hop DF link $x$-$z$-$y$ is given by
\begin{align}
R_{z} = \min \left ( \frac{1}{2} \log (1+ 2 \snr_{x,z} ) ,  \frac{1}{2}  \log  (1+ 2 \snr_{z,y} ) \right ) ,
\label{eq:24}
\end{align}
where $\snr_{x,y} \triangleq P_{r} (x,y) /(N_0 W)$ with $N_0$ and $W$ denoting the noise power spectral density and channel bandwidth respectively. Denote by $F_{R_z} (r)$ the cumulative distribution function (CDF) of the spectral efficiency $R_z$ of a typical two-hop DF link, and assume that $0 \leq R_z \leq \bar{r}$, where $\bar{r}$ is the largest achievable spectral efficiency.

We propose the following relay probing protocol for a typical communication pair: (i) the source-relay link is first probed, and (ii) if the source-relay link is blocked, the probing stops; otherwise, the relay-destination link is further probed. The physical meaning of probing a mmWave link may be understood as measuring the corresponding received signal-to-noise ratio (SNR): the link is considered blocked if the SNR is below certain threshold and not blocked otherwise.
To measure the received SNR, the transmitter and the receiver need to steer their antenna beams towards each other. We assume that each beamforming process requires a time duration $\tau$. The data communication occurs after the relay probing process and  lasts for a time duration $T$, after which a new relay probing process starts. Therefore, the communication process is periodic in time and each period consists of two steps: relay probing and data transmission. 

The relay probing is repeated independently across different periods. Note that the lengths of different  periods may differ as different numbers of relays may be probed in different periods. Further, when probing a relay $z$, whether the relay-destination $z$-$y$ link is probed or not depends on the probing outcome of the source-relay $x$-$z$ link. Thus, probing the relay $z$ consumes a (random) time duration $\tau ( 1 + \chi_{x,z} )$, and for a sample period in which $n$ relays are probed, the length of the period is given by
\begin{align}
T_n = \sum_{i=1}^n \tau ( 1 + \chi_{x,z} ) + T.
\end{align}
Figure \ref{fig:44} illustrates a sample realization of a period of relay probing and data transmission.

\begin{figure}
\centering
\includegraphics[width=8cm]{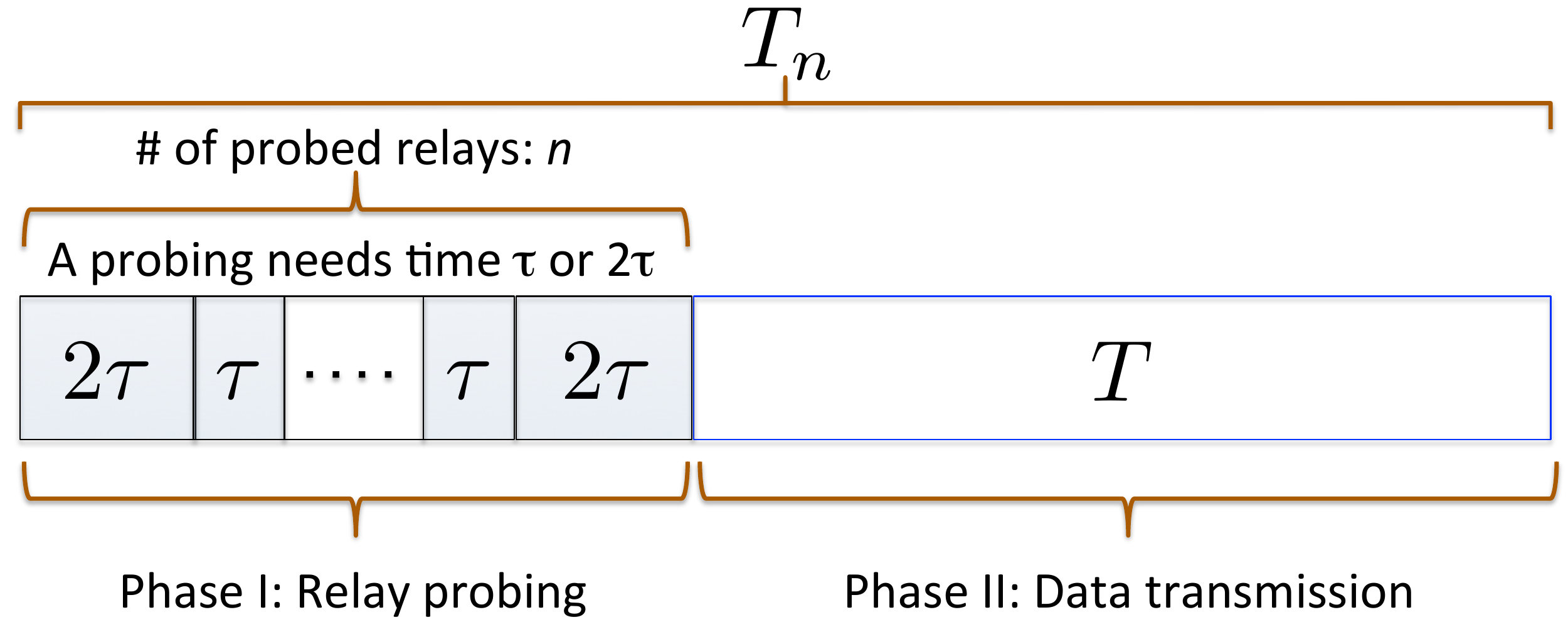}
\caption{A sample realization of relay probing and data transmission in a communication period.}
\label{fig:44}
\end{figure}


\subsection{Problem Formulation}
After probing $n$ relays, the source-destination pair can select the best relay for the two-hop communication. Accordingly,  the amount of bits that can be delivered equals
\begin{align}
U_n =W T \max ( R_{z_1},...,R_{z_n} ).
\end{align}
If we repeatedly use a stopping rule across $m$ different periods, the total number of bits that can be delivered equals $U_{N_1} + \cdots + U_{N_m}$, and the total amount of consumed time equals $T_{N_1} + \cdots T_{N_m}$, where $N_i$ denotes the number of probed relays in the $i$-th period. As a result, the average throughput (bit/s) equals $(U_{N_1} + \cdots + U_{N_m})/(T_{N_1} + \cdots T_{N_m})$. Dividing both the numerator and the denominator by $m$ and letting $m$ go to infinity, the last ratio converges to $\mathbb{E} [U_N]  / \mathbb{E} [T_N]$ by the law of large numbers. 

Our objective is to find an optimal stopping rule $N^\star$ which probes (in expectation) a finite number of relays and maximizes the average throughput:
\begin{align}
\mu^\star = \max_{N  \in \mathcal{C}} \frac{ \mathbb{E} [U_N]  }{ \mathbb{E} [T_N]  },
\label{eq:20}
\end{align}
where $\mathcal{C} = \{N \in \mathbb{N}: \mathbb{E} [T_N] < \infty  \}$ denotes the set of admissible stopping rules, and by definition $\mu^\star$ is the maximum throughput.

\section{Throughput-Optimal Relay Probing}
\label{sec:selection}

\subsection{The Associated Ordinary Optimal Stopping Problem}

It turns out very challenging to directly solve the throughput maximization problem (\ref{eq:20}). Nevertheless, we may solve (\ref{eq:20}) by considering the solution to the following ordinary optimal stopping problem:
\begin{align}
V( \mu ) \triangleq \max_{N \in \mathcal{C}}   \mathbb{E} [U_N - \mu T_N] .
\label{eq:21}
\end{align}
The following lemma proved in \cite{ferguson2012optimal} can be used to establish the relation between the solution to the throughput maximization problem (\ref{eq:20}) and the solution to the ordinary optimal stopping problem (\ref{eq:21}).  
\begin{lem}
$N^\star$ is an optimal stopping rule that attains the maximum throughput $\mu^\star$ in (\ref{eq:20}) if and only if $N^\star$ is an optimal stopping rule for the ordinary optimal stopping problem (\ref{eq:21}) with $\mu = \mu^\star$ and $V(\mu^\star)=0$.
\label{lem:1}
\end{lem}

Lemma \ref{lem:1} suggests the following way to solve the throughput maximization problem (\ref{eq:20}). First, for a given $\mu$, find the optimal stopping rule for the ordinary problem (\ref{eq:21}). Second, find $\mu^\star$ such that $V(\mu^\star) = 0$. Then the optimal stopping rule attaining $V(\mu^\star) = 0$ for 
the ordinary optimal stopping problem (\ref{eq:21}) is also an optimal stopping rule for the original throughput maximization problem (\ref{eq:20}).

For the ordinary optimal stopping problem (\ref{eq:21}), let us consider the following \textit{1-stage look-ahead} stopping rule:
\begin{align}
N_1(\mu) = \min \{ n \in \mathbb{N}: U_n - \mu T_n \geq \mathbb{E}[ U_{n+1} - \mu T_{n+1} |\mathcal{F}_n ]  \},
\label{eq:22}
\end{align}
where the subscript $_1$ and the parameter $\mu$ in $N_1(\mu)$ respectively indicate that the rule is only to look $1$ stage ahead and that the rule depends on $\mu$, and $\{\mathcal{F}_n\}$ is a filtration of the underlying probability space. 
The above 1-stage look-ahead stopping rule is myopic: it calls for stopping as long as the current utility is not less than the expected utility attained at the next stage. This myopic decision neglects the possibility that the expected utilities attained beyond the next stage may exceed the current utility, and thus in general is suboptimal. Somewhat surprisingly, this 1-stage look-ahead stopping rule turns out to be optimal for the problem in question, as stated in the following Lemma \ref{lem:3}.
\begin{lem}
Denote by $M_n = \max ( R_{z_1},...,R_{z_n} )$. For any $\mu > 0$, the 1-stage look-ahead stopping rule (\ref{eq:22}) is optimal. Further, it can be reduced to the following threshold rule:
\begin{align}
N_1(\mu) =  \min \{  n \in \mathbb{N}: M_n \geq \rho \},
\label{eq:25}
\end{align}
where $\rho$ is the unique root to the following equation:
\begin{align}
\mathbb{E} [ \max (R_z - \rho ,0) ] = \frac{\mu \tau}{WT} (1+p) .
\label{eq:18}
\end{align}
\label{lem:3}
\end{lem}
We omit the proof of Lemma \ref{lem:3} due to page constraints.

\subsection{Threshold Policy Achieves the Optimal Throughput}

Lemma \ref{lem:3} gives the the optimal stopping rule for the ordinary problem (\ref{eq:21}) with an arbitrary $\mu >0$. As noted in Lemma \ref{lem:1}, we are particularly interested in the $\mu^\star$ such that $V(\mu^\star) = 0$, because it would lead to the optimal stopping rule for the original throughput maximization problem (\ref{eq:20}). This $\mu^\star$ can be found by invoking Lemma \ref{lem:3}; the corresponding results are summarized in the following proposition.

\begin{pro}
The optimal stopping rule for the throughput maximization problem (\ref{eq:20}) is given by
\begin{align}
N^\star = \min \left \{ n \in \mathbb{N}: R_{z_n} \geq  \frac{\mu^\star}{W}  \right \},
\label{eq:27}
\end{align}
where $\mu^\star$ is the unique maximum throughput satisfying
\begin{align}
\mu^\star = \frac{\mathbb{E}[ \max( W T R_{z} - \mu^\star T,   0 ) ]}{\tau (1+p)}.
\label{eq:28}
\end{align}
Further, the optimal relay is the relay probed at stage $N^\star$.
\label{pro:3}
\end{pro}
\begin{proof}
We first claim that the rule (\ref{eq:25}) is equivalent to the following:
$
\hat{N}_1(\mu) = \min \{ n \in \mathbb{N}: R_{z_n} \geq \rho \}.
$
This can be shown by induction. Clearly, at stage $1$ the stopping rules $N_1(\mu)$ and $\hat{N}_1(\mu)$ are the same because $M_1 = R_{1}$. In particular, if $M_1 = R_{1} \geq \rho $, both $N_1(\mu)$ and $\hat{N}_1(\mu)$ call for stopping.  If $M_1 = R_{1} < \rho $, both $N_1(\mu)$ and $\hat{N}_1(\mu)$ call for continuing to stage $2$. At stage $2$, if $M_2 = \max (R_{z_1},R_{z_2}) \geq \rho$, then $M_1  = R_{z_2}$ because $R_{z_1} < \rho$ by induction. It follows that both $N_1(\mu)$ and $\hat{N}_1(\mu)$ call for stopping.  If $M_1 = \max (R_{z_1}, R_{z_2}) < \rho$, then $R_{z_2} < \rho$. Thus, both $N_1(\mu)$ and $\hat{N}_1(\mu)$ call for continuing to stage $3$. Repeating this argument for stages $3,4,...$, we can see that $N_1(\mu)$ and $\hat{N}_1(\mu)$ are indeed the same stopping rules.

Now we have shown that the optimal relay selection rule is to select the first relay satisfying $R_{z_n} \geq \rho$. In particular, we do not have to recall any of the previously probed relays $R_{z_m}, m < n$. Using this fact, we can see that the problem is invariant in time: If a relay is probed and we stop, the utility is $WTR_{z_1} - \mu \tau (1+\chi_{x,z_1})  - \mu T$; if we continue probing more relays, the utility is $V(\mu) - \mu \tau (1+\chi_{x,z_1}) $. Thus, the following optimality equation holds:
\begin{align}
V(\mu) = \mathbb{E} \big[ \max \big( & WTR_{z_1} - \mu \tau (1+\chi_{x,z_1})   
- \mu T,  V(\mu) - \mu \tau (1+\chi_{x,z_1}) \big ) \big] .
\end{align}
Since the optimal $\mu^\star$ is the one such that $V(\mu^\star) = 0$, the optimality equation reduces to
\begin{align}
0&= \mathbb{E}[ \max(WTR_{z_1} - \mu^\star \tau (1+\chi_{x,z_1})  - \mu^\star T,   - \mu^\star \tau (1+\chi_{x,z_1})   ) ] \notag \\
& = \mathbb{E}[ \max(WTR_{z_1} -   \mu^\star T,   0 ) ] -   \mu^\star \tau(1+p)  ,
\end{align}
from which we have
\begin{align}
\mathbb{E} \left[ \max(R_{z_1} -   \frac{\mu^\star}{W},   0 ) \right] = \frac{ \mu^\star \tau (1+p) }{WT}.
\label{eq:19}
\end{align}
Comparing (\ref{eq:18}) to (\ref{eq:19}), we conclude that $\rho =  \frac{\mu^\star}{W}$ and complete the proof.
\end{proof}

Prop. \ref{pro:3} implies that the optimal stopping rule is a pure threshold policy: the relay probing process stops once the achievable spectral efficiency of the currently probed two-hop link exceeds $\mu^\star/W$. In particular, the stopping rule is based on the state of the currently examined relay only, and it is not necessary to recall any previously probed relays: just select the relay probed at the stopping stage $N^\star$. This is a desirable feature in practical systems. In particular, due to the timing varying radio environment, the beamforming patterns trained for an earlier relay might become outdated or the two-hop link might enter deep fade. Choosing the most recently probed relay avoids such nuisances.

Note that it is challenging to derive a closed form solution for $\mu^\star$ from the fixed point equation (\ref{eq:28}) for a non-trivial spectral efficiency distribution $F_{R_z} (r)$. Instead, numerical methods are usually needed to compute $\mu^\star$. One possible numerical iterative algorithm is given as follows:
\begin{align}
\mu (t+1) = \frac{\mathbb{E}[ \max( W T R_{z} - \mu (t) T,   0 ) ]}{\tau (1+p)} ,
\label{eq:26}
\end{align}
where $t$ is the iteration index.
This iterative method is in essence a variation of Newton's method with all iterations using a unit step size. As shown in \cite{ferguson2012optimal}, for any non-negative initial value $\mu(0)$, the sequence $\{\mu(t)\}$ generated by the iteration (\ref{eq:26}) converges quadratically to $\mu^*$.

\section{Application: On/Off MmWave Links}

To obtain some insights and intuitions, we apply the derived analytical results to a homogeneous network with on/off mmWave links. Specifically, for the communication pair $(x,y)$ and any relay $z$, $P_r(x,z) = \chi_{x,z} P_1$ and  $P_r(z,y) = \chi_{z,y} P_2$, where $P_1$ and $P_2$ are some positive constants. Then the spectral efficiency of the two-hop link $x$-$z$-$y$ is given by
\[ R_z = \left\{ \begin{array}{ll}
         \bar{r} & \mbox{with probability $p^2$};\\
         0 & \mbox{with probability $1-p^2$}.\end{array} \right. \] 
Plugging the above into (\ref{eq:28}) and solving for $\mu^\star$ yields
\begin{align}
\mu^\star =   \frac{ W T p^2\bar{r} }{ (1+p)\tau + p^2 T }  .
\label{eq:30}
\end{align}
Accordingly, the optimal stopping rule is given by
\begin{align}
N^\star = \min \left \{ n \in \mathbb{N}: R_{z_n} \geq  \bar{r}  \left( 1 + \frac{1+p}{p^2} \frac{\tau}{T}  \right)^{-1}  \right \} .
\label{eq:31}
\end{align}

Several remarks are in order.

\textbf{Remark 1.} The numerator $W T p^2  \bar{r}$ in (\ref{eq:30}) is the expected number of bits that can be delivered in a communication period, while the denominator $(1+p)\tau + p^2 T$ in (\ref{eq:30}) is the expected duration of a communication period including relay probing and data transmission. Therefore, without a priori knowledge of the relays, the optimal stopping helps the system spend the right amount of resources on relay probing and yields the optimal throughput given  in (\ref{eq:30}).

\textbf{Remark 2.}
With a genie-aided relay selection, the system knows the right beamforming patterns and can directly use a two-hop link with spectral efficiency $\bar{r}$ and thus  the achievable throughput is $WT\bar{r}/T = W\bar{r}$. The ratio of the optimally stopped throughput and the genie-aided throughput is
\begin{align}
\frac{\mu^\star}{ W \bar{r} } = \frac{1}{1 + \frac{1+p}{p^2} \frac{\tau}{T}} .
\label{eq:32}
\end{align}
The gap between the optimally stopped throughput and the genie-aided throughput is solely determined by $\frac{\tau}{T}$ and $\frac{p^2}{1+p}$.
Intuitively, $\frac{\tau}{T}$ can be considered as normalized probing overhead, while $\frac{p^2}{1+p}$, which is increasing in $p \in (0,1]$, can be seen as a measure of the channel condition. As expected, the smaller the probing overhead or the better the channel condition, the smaller the probing overhead, the smaller the gap. In particular, if the mmWave links are always available, i.e., $p=1$, the ratio is $T/( T +2\tau )$, where the sole cost is $2\tau$ time resources spent on finding the right beamforming patterns.

\textbf{Remark 3.}
The optimal stopping rule (\ref{eq:31}) is a pure threshold policy: the relay probing process stops once $R_{z_n} / \bar{r} \geq \left( 1 + \frac{1+p}{p^2} \frac{\tau}{T}  \right)^{-1}$, which interestingly equals the ratio (\ref{eq:32}) of the optimally stopped throughput and the genie-aided throughput. Intuitively, the smaller the probing overhead or the better the channel condition, the higher the threshold, i.e., the system is more willing to probe more relays.

\textbf{Remark 4.}
Note that the threshold in (\ref{eq:27}) may not be unique, though the optimal throughput $\mu^\star$ is unique. Our current application with on/off links is one such example. Specifically, the spectral efficiency of any two-hop link $x$-$z$-$y$ is either $0$ or $\bar{z}$. As a result, any value in $(0, \bar{r}]$ can be used as a threshold for the optimal stopping.

\section{Simulation Results}

In this section, we provide simulation results to generate more insights into the derived theoretical results. The mmWave simulation setup closely follows \cite{khan2011millimeter} and is described as follows. The source is a pico BS located at $(-250 \textrm{ m}, 0)$ and the destination is a device located at $(250 \textrm{ m}, 0)$. The potential relays are other devices randomly distributed in the ball centered at $(0,0)$ with radius $250 \textrm{ m}$. The transmit powers of the pico BS and a device are respectively $30$ dBm and $23$ dBm. The beamforming (either transmit or receive) gains of the pico BS and a device are respectively $20$ dB and $10$ dB.
The channel bandwidth is $W=500$ MHz at a carrier frequency of $28$ GHz. The noise power spectral density is $-174$ dBm/Hz. The receiver noise figure is $7$ dB.
The pathloss model is $10\log_{10} \textrm{PL}( \| y - x \| )  = 141.3 + 20\log_{10} ( \| y - x \| /1000 )$ dB. Each link $x\to y$ is also subject to a random log-normal shadowing $\eta_{x,y}$ with $7$ dB deviation, and is available with probability $p$.
The data transmission phase lasts $T=1$ second. Note that the prototype in \cite{roh2014millimeter} requires $45$ ms for each beam searching and switching. Motivated by  this result, in our simulation we consider two beamforming time overhead values: $\tau=10$ ms and $50$ ms, indicating low overhead and high overhead, respectively.

In Figure \ref{fig:33}, we study how the throughput performance varies with the stopping spectral efficiency threshold $R_z$ under different probing overheads and link availabilities.  For each pair of $(\tau, p)$, Figure \ref{fig:33} clearly shows that there exists an optimal stopping threshold that achieves the maximum throughput. Further, the maximum throughput increases if the probing overhead $\tau$ decreases and/or the link availability $p$ increases, which is intuitive. On the one hand, for a fixed link availability $p$, the higher the  probing overhead $\tau$, the lower the optimal stopping threshold, i.e., less potential relaying devices are probed. On the other hand, for a fixed probing overhead $\tau$, the higher the link availability $p$, the higher the optimal stopping threshold, i.e., more potential relaying devices are explored.

\begin{figure}
\centering
\includegraphics[width=8cm]{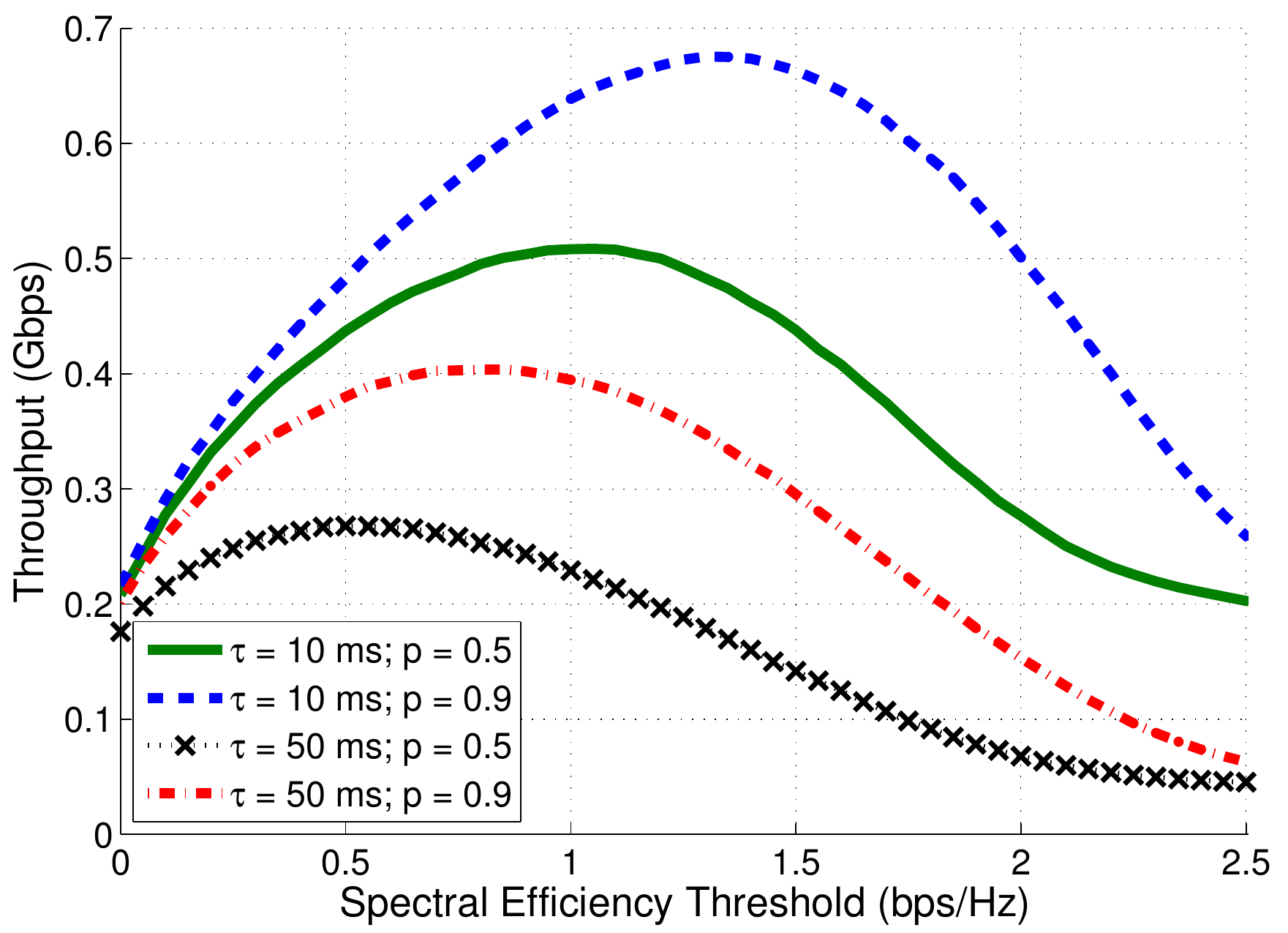}
\caption{Impact of probing overhead $\tau$ and link availability $p$ on the optimal stopping threshold.}
\label{fig:33}
\end{figure}

In Figure \ref{fig:55}, we compare the throughput performance attained by the derived optimal stopping rule to the performance of several heuristic relay probing strategies. In Figure \ref{fig:55}, the probing overhead $\tau=10$ ms, while similar observations hold for $\tau = 50$ ms. The first heuristic scheme is termed myopic stopping, in which the relay probing process is stopped at the first relay that can be used to establish both source-relay and relay-destination links. The second heuristic scheme is to probe a fixed number $\beta$ of relays and then selects the best one from the probed relays. We consider two values for $\beta$ in Figure \ref{fig:55}: $5$ and $10$. Figure \ref{fig:55} shows that, when the link availability is very low (e.g., $p=0.1$), myopic stopping is nearly optimal and outperforms the two fixed probing heuristics. This implies that, in a radio environment with many blockages, a mmWave system can stop probing relays immediately once it finds a feasible relay but still achieves close-to-optimum throughput performance. This myopic stopping scheme, however, performs poorly when the link availability $p$ increases. In particular, with a higher link availability $p$, it is valuable to explore more potential relays to find a good source-relay-destination route. As a result, the two fixed probing heuristics outperform the myopic stopping once the link availability $p$ becomes large enough. The derived optimal stopping approach outperforms the three heuristic schemes in then entire range of link availability $p$, and the throughput gains of the optimal stopping are remarkably large.

\begin{figure}
\centering
\includegraphics[width=8cm]{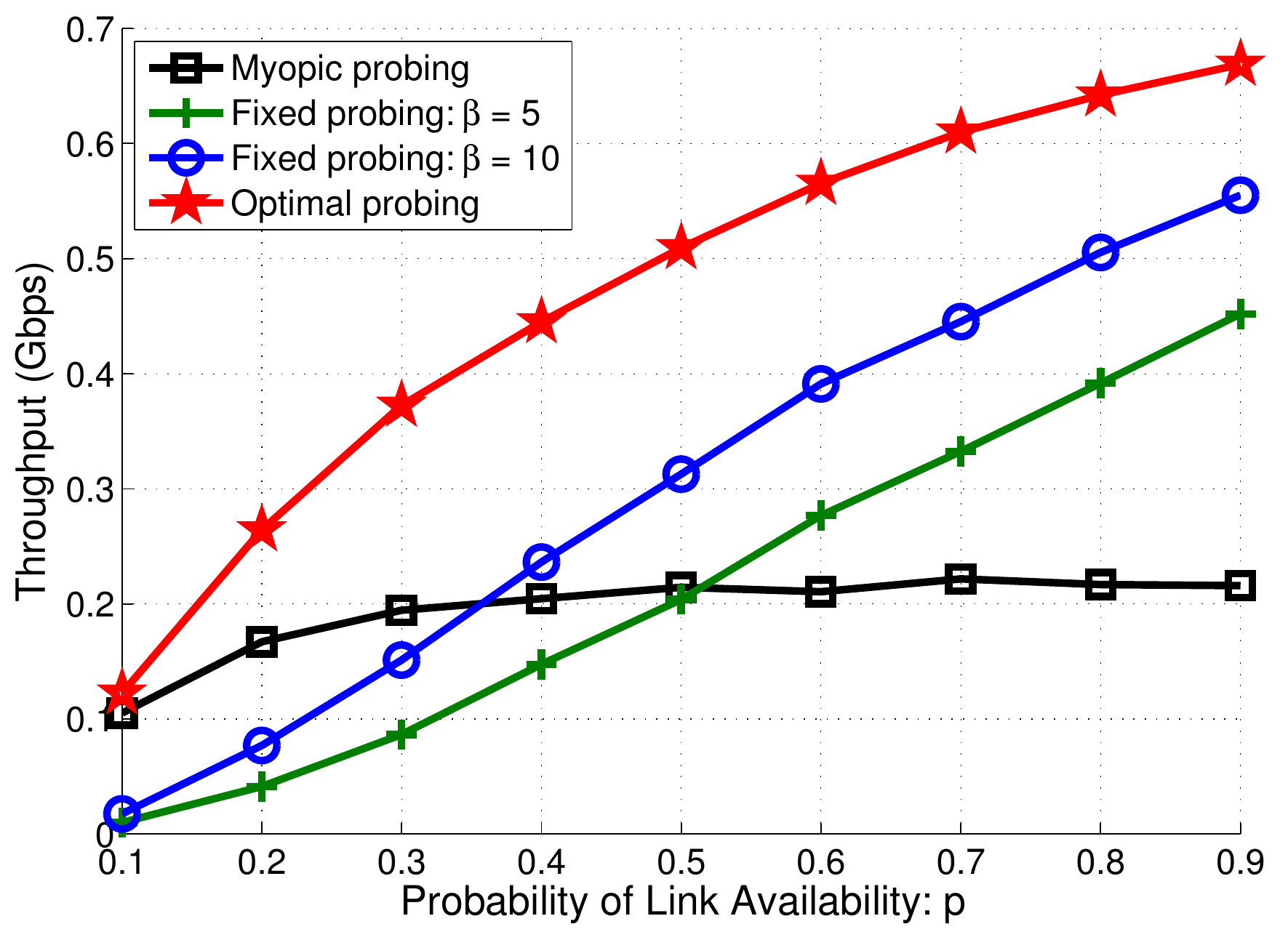}
\caption{Comparison of throughput performance under different relay probing strategies.}
\label{fig:55}
\end{figure}

\section{Conclusions}

In bandwidth-limited cellular systems, the gains of multi-hop relaying may be modest since the loss in the degrees of freedom (due to half-duplex constraints) eats into the gains. In contrast, mmWave systems have many degrees of freedom (in terms of bandwidth) but are power-limited and susceptible to blockages. Using infrastructure-based and/or D2D relaying can help increase the connectivity and range of  mmWave cellular systems. In this correspondence, we have taken some initial steps towards studying relaying probing and selection in a two-hop DF mmWave cellular system. We have found that a threshold-based policy can optimally balance the tradeoff between the throughput gain from searching a better relay and the throughput loss due to higher relay probing overhead. Numerical results have demonstrated the remarkable throughput gains of the threshold-based relay probing policy (versus several heuristic schemes). Future work may consider a heterogeneous scenario where the mmWave links are subject to non-i.i.d. Bernoulli blockages. It is also of interest to extend the current work to multiuser scenarios where many source-destination pairs exist.

\bibliographystyle{IEEEtran}
\bibliography{IEEEabrv,Reference}

\begin{thebibliography}{10}
\providecommand{\url}[1]{#1}
\csname url@samestyle\endcsname
\providecommand{\newblock}{\relax}
\providecommand{\bibinfo}[2]{#2}
\providecommand{\BIBentrySTDinterwordspacing}{\spaceskip=0pt\relax}
\providecommand{\BIBentryALTinterwordstretchfactor}{4}
\providecommand{\BIBentryALTinterwordspacing}{\spaceskip=\fontdimen2\font plus
\BIBentryALTinterwordstretchfactor\fontdimen3\font minus
  \fontdimen4\font\relax}
\providecommand{\BIBforeignlanguage}[2]{{%
\expandafter\ifx\csname l@#1\endcsname\relax
\typeout{** WARNING: IEEEtran.bst: No hyphenation pattern has been}%
\typeout{** loaded for the language `#1'. Using the pattern for}%
\typeout{** the default language instead.}%
\else
\language=\csname l@#1\endcsname
\fi
#2}}
\providecommand{\BIBdecl}{\relax}
\BIBdecl

\bibitem{rappaport2013millimeter}
T.~Rappaport, S.~Sun, R.~Mayzus, H.~Zhao, Y.~Azar, K.~Wang, G.~Wong, J.~Schulz,
  M.~Samimi, and F.~Gutierrez, ``Millimeter wave mobile communications for {5G}
  cellular: It will work!'' \emph{IEEE Access}, vol.~1, pp. 335--349, May 2013.

\bibitem{khan2011millimeter}
F.~Khan and J.~Pi, ``Millimeter--wave mobile broadband: {Unleashing 3--300GHz
  spectrum},'' in \emph{Proceedings of IEEE WCNC}, 2011.

\bibitem{roh2014millimeter}
W.~Roh, J.-Y. Seol, J.~Park, B.~Lee, J.~Lee, Y.~Kim, J.~Cho, K.~Cheun, and
  F.~Aryanfar, ``Millimeter-wave beamforming as an enabling technology for {5G}
  cellular communications: {Theoretical} feasibility and prototype results.''
  \emph{IEEE Communications Magazine}, vol.~52, no.~2, pp. 106--113, February
  2014.

\bibitem{andrews2014will}
J.~G. Andrews, S.~Buzzi, W.~Choi, S.~Hanly, A.~Lozano, A.~Soong, and J.~Zhang,
  ``What will {5G} be?'' \emph{IEEE Journal on Selected Areas in
  Communications}, vol.~32, no.~6, pp. 1065--1082, June 2014.

\bibitem{rangan2014millimeter}
S.~Rangan, T.~Rappaport, and E.~Erkip, ``Millimeter-wave cellular wireless
  networks: Potentials and challenges,'' \emph{Proceedings of the IEEE}, vol.
  102, no.~3, pp. 366--385, March 2014.

\bibitem{singh2009blockage}
S.~Singh, F.~Ziliotto, U.~Madhow, E.~Belding, and M.~Rodwell, ``Blockage and
  directivity in {60GHz} wireless personal area networks: From cross-layer
  model to multihop {MAC} design,'' \emph{IEEE Journal on Selected Areas in
  Communications}, vol.~27, no.~8, pp. 1400--1413, October 2009.

\bibitem{bai2014coverage}
T.~Bai and R.~Heath, ``Coverage and rate analysis for millimeter-wave cellular
  networks,'' \emph{IEEE Transactions on Wireless Communications}, vol.~14,
  no.~2, pp. 1100--1114, February 2015.

\bibitem{lu2012modeling}
J.~Lu, D.~Steinbach, P.~Cabrol, and Pietraski, ``Modeling the impact of human
  blockers in millimeter wave radio links,'' \emph{ZTE Communications
  Magazine}, vol.~10, no.~4, pp. 23--28, December 2012.

\bibitem{lin2014connectivity}
X.~Lin and J.~Andrews, ``Connectivity of millimeter wave networks with
  multi-hop relaying,'' \emph{IEEE Wireless Communications Letters}, vol.~4,
  no.~2, pp. 209--212, April 2015.

\bibitem{lin2013overview}
X.~Lin, J.~G. Andrews, A.~Ghosh, and R.~Ratasuk, ``An overview of {3GPP}
  device-to-device proximity services,'' \emph{IEEE Communications Magazine},
  vol.~52, no.~4, pp. 40--48, April 2014.

\bibitem{pabst2004relay}
R.~Pabst, B.~H. Walke, D.~C. Schultz, P.~Herhold, H.~Yanikomeroglu,
  S.~Mukherjee, H.~Viswanathan, M.~Lott, W.~Zirwas, M.~Dohler \emph{et~al.},
  ``Relay-based deployment concepts for wireless and mobile broadband radio,''
  \emph{IEEE Communications Magazine}, vol.~42, no.~9, pp. 80--89, September
  2004.

\bibitem{le2007multihop}
L.~Le and E.~Hossain, ``Multihop cellular networks: {Potential} gains, research
  challenges, and a resource allocation framework,'' \emph{IEEE Communications
  Magazine}, vol.~45, no.~9, pp. 66--73, September 2007.

\bibitem{gozalvez2013experimental}
J.~Gozalvez and B.~Coll-Perales, ``Experimental evaluation of multihop cellular
  networks using mobile relays,'' \emph{IEEE Communications Magazine}, vol.~51,
  no.~7, pp. 122--129, July 2013.

\bibitem{lin2014multihop}
Y.-D. Lin, Y.-C. Hsu, M.~Chattterjee, and T.~Kunz, ``Multihop cellular: {From}
  research to systems, standards, and applications,'' \emph{IEEE Wireless
  Communications}, vol.~21, no.~5, pp. 12--13, October 2014.

\bibitem{drucker1988development}
E.~H. Drucker, ``Development and application of a cellular repeater,'' in
  \emph{Proceedings of IEEE VTC}, 1988, pp. 321--325.

\bibitem{laneman2004cooperative}
J.~N. Laneman, D.~N. Tse, and G.~W. Wornell, ``Cooperative diversity in
  wireless networks: Efficient protocols and outage behavior,'' \emph{IEEE
  Transactions on Information Theory}, vol.~50, no.~12, pp. 3062--3080,
  December 2004.

\bibitem{kramer2005cooperative}
G.~Kramer, M.~Gastpar, and P.~Gupta, ``Cooperative strategies and capacity
  theorems for relay networks,'' \emph{IEEE Transactions on Information
  Theory}, vol.~51, no.~9, pp. 3037--3063, September 2005.

\bibitem{chen2006modulation}
D.~Chen and J.~N. Laneman, ``Modulation and demodulation for cooperative
  diversity in wireless systems,'' \emph{IEEE Transactions on Wireless
  Communications}, vol.~5, no.~7, pp. 1785--1794, July 2006.

\bibitem{shan2009distributed}
H.~Shan, W.~Zhuang, and Z.~Wang, ``Distributed cooperative {MAC} for multihop
  wireless networks,'' \emph{IEEE Communications Magazine}, vol.~47, no.~2, pp.
  126--133, 2009.

\bibitem{bletsas2006simple}
A.~Bletsas, A.~Khisti, D.~P. Reed, and A.~Lippman, ``A simple cooperative
  diversity method based on network path selection,'' \emph{IEEE Journal on
  Selected Areas in Communications}, vol.~24, no.~3, pp. 659--672, March 2006.

\bibitem{zhou2011link}
Y.~Zhou, J.~Liu, L.~Zheng, C.~Zhai, and H.~Chen, ``Link-utility-based
  cooperative {MAC} protocol for wireless multi-hop networks,'' \emph{IEEE
  Transactions on Wireless Communications}, vol.~10, no.~3, pp. 995--1005,
  March 2011.

\bibitem{zhai2009lifetime}
C.~Zhai, J.~Liu, L.~Zheng, and H.~Xu, ``Lifetime maximization via a new
  cooperative {MAC} protocol in wireless sensor networks,'' in
  \emph{Proceedings of IEEE GLOBECOM}, 2009, pp. 1--6.

\bibitem{marchenko2009selecting}
N.~Marchenko, E.~Yanmaz, H.~Adam, and C.~Bettstetter, ``Selecting a spatially
  efficient cooperative relay,'' in \emph{Proceedings of IEEE GLOBECOM}, 2009,
  pp. 1--7.

\bibitem{wang2009distributed}
B.~Wang, Z.~Han, and K.~J.~R. Liu, ``Distributed relay selection and power
  control for multiuser cooperative communication networks using {Stackelberg}
  game,'' \emph{IEEE Transactions on Mobile Computing}, vol.~8, no.~7, pp.
  975--990, July 2009.

\bibitem{ng2007joint}
T.~C.-Y. Ng and W.~Yu, ``Joint optimization of relay strategies and resource
  allocations in cooperative cellular networks,'' \emph{IEEE Journal on
  Selected Areas in Communications}, vol.~25, no.~2, pp. 328--339, February
  2007.

\bibitem{fan2007mimo}
Y.~Fan and J.~Thompson, ``{MIMO} configurations for relay channels: {Theory}
  and practice,'' \emph{IEEE Transactions on Wireless Communications}, vol.~6,
  no.~5, pp. 1774--1786, May 2007.

\bibitem{yang2009relay}
Y.~Yang, H.~Hu, J.~Xu, and G.~Mao, ``Relay technologies for {WiMAX and
  LTE}-advanced mobile systems,'' \emph{IEEE Communications Magazine}, vol.~47,
  no.~10, pp. 100--105, October 2009.

\bibitem{hur2013millimeter}
S.~Hur, T.~Kim, D.~Love, J.~Krogmeier, T.~Thomas, and A.~Ghosh, ``Millimeter
  wave beamforming for wireless backhaul and access in small cell networks,''
  \emph{IEEE Transactions on Communications}, vol.~61, no.~10, pp. 4391--4403,
  October 2013.

\bibitem{singh2014tractable}
S.~Singh, M.~N. Kulkarni, A.~Ghosh, and J.~G. Andrews, ``Tractable model for
  rate in self-backhauled millimeter wave cellular networks,'' \emph{IEEE
  Journal on Selected Areas in Communications}, to appear. Available at arXiv
  preprint arXiv:1407.5537.

\bibitem{ferguson2012optimal}
T.~S. Ferguson, ``Optimal stopping and applications,'' 2012. Available at
  http://www.math.ucla.edu/~tom/Stopping/Contents.html.

\end{thebibliography}

\end{document}